\documentclass{amsart}
\usepackage{amsmath}%
\usepackage{amsfonts}%
\usepackage{amssymb}%
\usepackage{graphicx}
\newtheorem{theorem}{Theorem}
\theoremstyle{plain}

\newtheorem{lemma}{Lemma}

\newtheorem{preremark}{Remark}

\numberwithin{equation}{section}
\begin{document}
\title[Classical Solvability for Vlasov-Maxwell]{Classical solvability of the Relativistic Vlasov-Maxwell system with bounded spatial density}
\author{Reinel Sospedra-Alfonso}
\author{Reinhard Illner}
\address
{Department of Mathematics and Statistics,\newline%
\indent University of Victoria, PO BOX 3045 STN CSC, Victoria BC V8W 3P4.}%
\email[R. Sospedra-Alfonso]{sospedra@math.uvic.ca}%
\email[R. Illner]{rillner@math.uvic.ca}%

\date{Jun 24, 2009}

\begin{abstract}
In \cite{GS1}, Glassey and Strauss showed that if the growth in the momentum of the particles is controlled, then the relativistic Vlasov-Maxwell system has a classical solution globally in time. Later they proved that such control is achieved if the \textit{kinetic energy density} of the particles remains bounded for all time \cite{GS2}. Here, we show that the latter assumption can be weakened to the boundedness of the \textit{spatial density}. 
\end{abstract}
\maketitle

\section{Introduction}
The relativistic Vlasov-Maxwell (RVM) system describes the time evolution of a collisionless plasma whose particles interact through the self-induced electromagnetic field. The plasma is assumed to be at high temperatures, thus the particles may travel at speeds comparable to the speed of light. For a single species, the model equations are  
\begin{equation}
\partial_tf+v\cdot\nabla_{x}f+\left(E+v\times B\right)\cdot\nabla_{p}f=0
\label{Vlasov Equation}
\end{equation}
\begin{equation}
\partial_tE-\nabla\times B=-j
\label{Maxwell Evolution 1}
\end{equation}
\begin{equation}
\partial_tB+\nabla\times E=0
\label{Maxwell Evolution 2}
\end{equation}
\begin{equation}
 \begin{tabular}{cc}
    $\nabla\cdot E=\rho$, & $\nabla\cdot B=0$.
 \end{tabular}
\label{Maxwell Constraints}
\end{equation}
We have set to unity the mass and charge of the particles as well as the speed of light. The density 
function $f=f(t,x,p)$ depends on time $t\in(0,\;\infty)$, position $x\in\mathbb{R}^3$ and momentum $p\in\mathbb{R}^3$. $E=E(t,x)$ and $B=B(t,x)$ denote the electric and magnetic fields, respectively, and $v$ stands for the relativistic velocity, i.e.,
\begin{displaymath}
v=\frac{p}{\sqrt{1+\left|p\right|^2}}.
\end{displaymath}

The system (\ref{Vlasov Equation}-\ref{Maxwell Constraints}) is coupled via the charge (or spatial) 
density $\rho=\rho(t,x),$ and the current density $j=j(t,x),$  given by 
\begin{equation}
 \begin{tabular}{cc}
  $j=\int_{\mathbb{R}^3}vfdp$, & $\rho=\int_{\mathbb{R}^3}fdp$.
 \end{tabular}
\label{Density and Current}
\end{equation}

The Cauchy problem to the RVM system is (\ref{Vlasov Equation}-\ref{Maxwell Constraints}) with initial data
\begin{equation}
 \begin{tabular}{ccc}
   $f_{|t=0}=f_0$, & $E_{|t=0}=E_0$, & $B_{|t=0}=B_0$
 \end{tabular}
\label{Initial Data} 
\end{equation}
satisfying the constraints (\ref{Maxwell Constraints}).

The global in time classical solvability for this system remains an open problem. On the other hand, the existence of local solutions was first proved in \cite{Wollman} for $f_0$ having compact support. The key result on global existence is due to Glassey and Strauss and can be found in \cite{GS1}. They showed that local solutions may be extended globally in time if the momentum support of $f$ is controlled for all $t$. Hence the implication that, for smooth initial data, a singularity could occur only if some particles travel at speeds arbitrarily close to the speed of light. Two different versions of this result are given in \cite{KlainermanStafillani,Bouchut}. An extensive review on the RVM system is provided in the monograph \cite{Glassey Book}.   

In \cite{GS2}, Glassey and Strauss weakened the assumption made in \cite{GS1}. They showed that if the kinetic energy density of the particles (defined in (\ref{Kinetic Energy Particles}) below) is assumed to be bounded for all times, then the problem can be reduced to the situation studied in \cite{GS1}, i.e., to the boundedness of the momentum support of $f$. Recently, Pallard improved 
this result for compactly supported initial data. In \cite{Pallard}, he showed that if 
\begin{equation}
\label{Lq Moment}
\sup_{0\leq t\leq T}\left\|\int_{\mathbb{R}^3}\left(1+\left|p\right|^2\right)^{\frac{\theta}{2}}
f(t,\cdot,p)dp\right\|_{L^q_x}<\infty,
\end{equation}
with $6\leq q\leq \infty$ and $\theta>4/q$ (strict), then the $p$-support of $f$ is uniformly bounded for all $0\leq t\leq T$. The assumption made in \cite{GS2} corresponds to the case $(\theta=1,q=\infty)$. Here, we weaken the assumption in \cite{Pallard} a little bit more: we prove that if the spatial density of the particles remains bounded for all times, i.e., if (\ref{Lq Moment}) holds for $(\theta=0,q=\infty)$, then the problem can be reduced to that in \cite{GS1} and therefore, the existence of classical solutions globally in time follows. 


\section{Preliminaries}
\label{Pre}

Consider a $C^1$ solution of the system (\ref{Vlasov Equation}-\ref{Density and Current}). We abbreviate the Lorentz force by $K=E+v\times B$. It clearly satisfies an estimate $\left|K\right|\leq\left|E\right|+\left|B\right|=:\bar{\left|K\right|}$ since $|v|\leq1$. Further, let $\omega=\left(x-y\right)/\left|x-y\right|$. Following \cite{GS1} and \cite{GS2}, we represent the components of the electric field by

\begin{equation}
\label{Full Representation}
  4\pi E^i(t,x)=(E^i)_0(t,x)+E^i_T(t,x)+E^i_S(t,x),
\end{equation}
where $(E^i)_0$ depends on the initial data and  

\begin{equation}
\label{E_S}  
E^i_T(t,x)=\int_{\left|x-y\right|\leq t}\int_{\mathbb{R}^3}a^i(\omega,v)f(t-\left|x-y\right|,y,p)dp
\frac{dy}{\left|x-y\right|^2},
\end{equation}

\begin{equation}
\label{E_T}  
E^i_S(t,x)=\int_{\left|x-y\right|\leq t}\int_{\mathbb{R}^3}\nabla_pb^i(\omega,v)
\cdot(Kf)(t-\left|x-y\right|,y,p)dp\frac{dy}{\left|x-y\right|},
\end{equation}
with kernels $a^i(\omega,v)$ and $b^i(\omega,v)$ satisfying the estimate 

\begin{equation}
\label{Estimates Kernels}
  \left|a^i(\omega,v)\right|+\left|\nabla_pb^i(\omega,v)\right|\leq c\sqrt{1+\left|p\right|^2}.
\end{equation}

If we denote the kinetic energy density\footnote{Actually, this is the \textit{mechanical} 
energy density. The kinetic energy differs from the mechanical energy by the energy of the particles at rest, which is a constant. Nevertheless, we follow the previous convention.} of the particles by
\begin{equation}
\label{Kinetic Energy Particles}
  h(t,x)=\int_{\mathbb{R}^3}\sqrt{1+\left|p\right|^2}f(t,x,p)dp, 
\end{equation}
and use the substitution $y=x-\omega\left(t-s\right)$ in (\ref{E_S}) and (\ref{E_T}), we see that 
\begin{equation}
\label{Estimate E_S}  
\left|E_T(t,x)\right|\leq c\int^t_0\int_{\left|\omega\right|=1}h(s,x-\omega\left(t-s\right))d\omega ds
\end{equation}
and
\begin{equation}
\label{Estimate E_T}  
\left|E_S(t,x)\right|\leq c\int^t_0 \left(t-s\right)\int_{\left|\omega\right|=1}
(h\bar{\left|K\right|})(s,x-\omega\left(t-s\right))d\omega ds.
\end{equation}

Similarly, we can represent the magnetic field as in (\ref{Full Representation}-\ref{E_T}), with kernels $a^i(\omega,v)$ and $b^i(\omega,v)$ that can be bounded as in (\ref{Estimates Kernels}). The estimates (\ref{Estimate E_S}) and (\ref{Estimate E_T}) remain valid for $B_S$ and $B_T$ as well. Now, if we gather all these estimates and use the representation (\ref{Full Representation}) for $E$ and the corresponding representation for $B$, we observe that  
\begin{equation}
\label{Energy Hypothesis}
  \sup_{0\leq t\leq T}\sup_{x\in\mathbb{R}^3}h(t,x)\leq c_T
\end{equation}
implies 
\begin{equation}
\label{Gronwall Field}
  \sup_{x\in\mathbb{R}^3}\left(\left|E(t,x)\right|+\left|B(t,x)\right|\right)\leq c_T\left\{1+\int^t_0\sup_{x\in\mathbb{R}^3}\left(\left|E(s,x)\right|+\left|B(s,x)\right|\right)ds\right\}.
\end{equation}
The Gronwall lemma applies and the field remains bounded for all $t\leq T$.

On the other hand, the characteristics associated to the equation (\ref{Vlasov Equation}) are solutions of
 \begin{eqnarray}
 \label{Characteristics X}
 \dot{X}(t,x_0,p_0) & = & V(P(t,x_0,p_0))\\ 
 \dot{P}(t,x_0,p_0) & = & K(t,X(t,x_0,p_0),P(t,x_0,p_0)),\
 \label{Characteristics P}
 \end{eqnarray}
with $(X(0,x_0,p_0),P(0,x_0,p_0))\equiv(x_0,p_0)\in\text{supp}f_0$. Thus,
 \begin{equation}
 \label{Estimate Characteristics}
   \left|P(t,\cdot)\right|\leq\left|p_0\right|+\int^t_0\left(\left|E(s,X(s,\cdot))\right|+\left|B(s,X(s,\cdot))\right|\right)ds.
 \end{equation}

Hence, the boundedness of the field implies the boundedness of the momentum support of $f$, as long as $f_0$ has compact support in $p$. Now, if we apply the above reasoning to the approximate sequence of solutions $\left\{(f^n,E^n,B^n)\right\}$ introduced in \cite{GS1} (and which we will present in the last section), we have that the $p$-supports of $\left\{f^n\right\}$ are uniformly bounded in $n$. But this guarantees the global existence and uniqueness of classical solutions to the system (\ref{Vlasov Equation}-\ref{Initial Data}), as stated in \cite[Theorem 1]{GS1}. Therefore, the assumption ($\ref{Energy Hypothesis}$) (as well as for its approximations $h^n$, uniformly in $n$) implies the existence and uniqueness of classical solutions to the RVM system for all time. This is the result proved in \cite{GS2}. In this article we show that the condition ($\ref{Energy Hypothesis}$) can be relaxed to the boundedness of the spatial density. In detail, we prove

\begin{theorem} 
Let $f_0\in C^1_0(\mathbb{R}^6)$ and $(E_0,B_0)\in C^2_0(\mathbb{R}^3)$ satisfy the constraints (\ref{Maxwell Constraints}). For $T>0$, assume that  
\begin{displaymath}
   \sup_{0\leq t\leq T}\sup_{x\in\mathbb{R}^3}\int_{\mathbb{R}^3}f(t,x,p)dp\leq c_T,
\end{displaymath}
as well as for its approximations $f^n$ defined in (\ref{Vlasov Recursive}), uniformly in $n$. Then, there exists a unique $C^1$ solution for all time of the corresponding Cauchy problem to the RVM system.
\end{theorem}

Our proof relies on two key observations. The first one is that we could avoid the uniform estimates on the field used to bound the $p$-support of $f$ in (\ref{Estimate Characteristics}), (as, for instance, the one that derives from (\ref{Gronwall Field})), if instead we use estimates on the time-integral of the field along characteristics. This was already noticed in \cite{Pallard} and here we use modifications to some of the estimates provided there. Estimates of this type were successfully used in \cite{Calogero} as well.  

As for the second observation, define
\begin{displaymath}
  W(p)=\sqrt{1+\left|p\right|^2}.
\end{displaymath}
Clearly, $\dot{W}(p)=v\cdot\dot{p}$. Thus, since $v\cdot\left(v\times B\right)\equiv0$, along the characteristics solving (\ref{Characteristics X}-\ref{Characteristics P}) we have that  
 \begin{equation}
 \dot{W}(P(t,x_0,p_0)) = V(P(t,x_0,p_0))\cdot E(t,X(t,x_0,p_0)).
 \label{Characteristics Like W}
 \end{equation}
Since $\left|v\right|\leq1$, we obtain
\begin{equation}
\label{Estimate Characteristics Like}
  W(t)\leq\ W(0)+\int^t_0\left| E(s,X(s))\right|ds,
 \end{equation}
where we have abbreviated $W(t)=W(P(t,\cdot))$ and $E(s,X(s))=E(s,X(s,\cdot))$, respectively. 

Our strategy is to reduce (\ref{Estimate Characteristics Like}) to a Gronwall-type inequality by using estimates on the time-integral of the field along characteristics mentioned above. In turn, this will provide the uniform bound on the $p$-support of $f$ that guarantee the existence of classical solutions to the RVM system for all time. We remark that (\ref{Characteristics Like W}) arises naturally in the covariant formulation of the electrodynamics of the relativistic particle \cite[Chapter 12]{Jackson}. The quantity $W(p)$ is the kinetic energy of a single particle, and the set of equations (\ref{Characteristics P}) and (\ref{Characteristics Like W}) describes the general motion of a charged particle in an external electromagnetic field. For the present case that field is induced by the remaining charges of the system and it is computed by means of the Maxwell equations (\ref{Maxwell Evolution 1}-{\ref{Maxwell Constraints}}).

  
\section{Bounding $W(p)$}

Let $(f,E,B)$ be a $C^1$ solution of the system (\ref{Vlasov Equation}-\ref{Density and Current}). Define 
\begin{eqnarray}
  \bar{p}(t) & = & \sup\left\{\left|p\right|:\exists 0\leq s\leq t, x\in\mathbb{R}^3:f(s,x,p)\neq0\right\}\nonumber\\
  & \equiv & \sup\left\{\left|P(s,x_0,p_0)\right|:0\leq s\leq t,(x_0,p_0)\in\texttt{supp} f_0\right\},\nonumber
\end{eqnarray}
and set 
\begin{equation}
\label{Definitions}
\bar{v}(t)=\frac{\bar{p}(t)}{\sqrt{1+\bar{p}^2(t)}},\hspace{1.0cm} \bar{W}(t)=\sqrt{1+\bar{p}^2(t)}. 
\end{equation}
Throughout this section, we assume that for each $0\leq t\leq T$ and $(x_0,p_0)\in\texttt{supp} f_0$
\begin{equation}
\label{Condition Finite Momentum}
  \sup_{0\leq s\leq t}|\dot{X}(s,x_0,p_0)|\leq\bar{v}(t)<1,
\end{equation}
and that  
\begin{equation}
  \label{Assumption 2}
  \sup_{0\leq t\leq T}\sup_{x\in\mathbb{R}^3}\int_{\mathbb{R}^3}f(t,x,p)dp\leq c_T.
\end{equation}
The assumption 
(\ref{Assumption 2}) readily implies
\begin{equation}
\label{W Bounds h}
  h(t,x)\leq\int_{\left|p\right|\leq\bar{p}(t)}\sqrt{1+\left|p\right|^2}f(t,x,p)dp\leq c_T\bar{W}(t).
\end{equation}

Less straightforward is the following implication:
\begin{lemma} 
\label{Poynting}
If (\ref{Assumption 2}) holds, then
 \begin{equation}
   \label{Assumption 3}
   \sup_{0\leq t\leq  T}\int_{\mathbb{R}^3}\left(\left|E(t,x)\right|^2+\left|B(t,x)\right|^2\right)dx\leq c_T.
 \end{equation}
\end{lemma}
\begin{proof} The structure of the Vlasov equation implies that $f$ is constant along characteristics. Therefore, all its $L^q$-norms are preserved in time. In particular, we have that $\left\|f_0\right\|_{L^1_{x,p}}=\left\|f(t)\right\|_{L^1_{x,p}}\equiv\left\|\rho(t)\right\|_{L^1_x}$, $0\leq t\leq T$. Since $\left|v\right|\leq1$, then $\left|j\right|\leq\rho$ and therefore    $$\left\|j(t)\right\|^2_{L^{2}_x}\leq\left\|\rho(t)\right\|_{L^{\infty}_x}\left\|\rho(t)\right\|_{L^{1}_x}\leq c_T.$$
 
On the other hand, if we multiply (\ref{Maxwell Evolution 1}) and (\ref{Maxwell Evolution 2}) by $E$ and $B$ respectively, add up the resultant equations and integrate over $\mathbb{R}^3$, we find that
\begin{equation}
\label{Uno}
\frac{1}{2}\frac{d}{dt}\int_{\mathbb{R}^3}\left(\left|E(t,x)\right|^2+\left|B(t,x)\right|^2\right)dx=-\int_{\mathbb{R}^3}j(t,x)\cdot E(t,x)dx.
\end{equation}
The Cauchy-Schwarz and Young inequalities, together with the above estimate on the current density yield
\begin{equation}
\label{Dos}
\left|\int_{\mathbb{R}^3}j(t,x)\cdot E(t,x)dx\right|\leq c_T+\frac{1}{2}\int_{\mathbb{R}^3}\left(\left|E(t,x)\right|^2+\left|B(t,x)\right|^2\right)dx.
\end{equation}
Hence, we combine (\ref{Uno}) with (\ref{Dos}) and invoke the Gronwall lemma to conclude that the uniform estimate (\ref{Assumption 3}) indeed holds.  
\end{proof}

We notice that although a $C^1$ solution to (\ref{Vlasov Equation}-\ref{Density and Current}) preserves the total energy in time \cite{Glassey Book} (i.e., the quantity 
\begin{displaymath}
\int\int_{\mathbb{R}^3\times\mathbb{R}^3}\sqrt{1+\left|p\right|^2}f(t,x,p)dxdp+\frac{1}{8\pi}
\int_{\mathbb{R}^3}\left(\left|E(t,x)\right|^2+\left|B(t,x)\right|^2\right)dx
\end{displaymath}
does not depend on $t$, which in turn implies (\ref{Assumption 3})), the approximate solutions used in the next section may not. Nevertheless, they do satisfy (\ref{Uno}), and so we can use the previous lemma instead.    

Now, for an arbitrary non-negative function $g$, define the integrals
 \begin{eqnarray}
I_k(g;t) & = & \int^t_0\int^s_0(s-\sigma)^k\int_{\left|\omega\right|=1}g(\sigma,X(s)-\omega(s-\sigma))
d\omega d\sigma ds\nonumber\\ 
         & = & \int^t_0\int^t_{\sigma}(s-\sigma)^k\int_{\left|\omega\right|=1}g(\sigma,X(s)-\omega(s-\sigma))
         d\omega dsd\sigma\nonumber,
 \end{eqnarray}
with $k=0,1$. We further define $\mathcal{I}_k(g;\sigma,t)$ as
\begin{equation}
\label{Pallard Integral}
\mathcal{I}_k(g;\sigma,t)=\int^t_{\sigma}(s-\sigma)^k\int_{\left|\omega\right|=1}g(\sigma,X(s)-\omega(s-\sigma))d\omega ds
\end{equation}
such that
\begin{displaymath}
   I_k(g;t)=\int^t_0\mathcal{I}_k(g;\sigma,t)d\sigma.
\end{displaymath}

If we use the representation of the electric field introduced in the previous section and apply the 
estimates (\ref{Estimate E_S}) and (\ref{Estimate E_T}) to (\ref{Estimate Characteristics Like}), we find that 

\begin{equation}
\label{Working Equation}
   W(t)\leq W(0)+c_T+cI_0(h;t)+cI_1(h\bar{\left|K\right|};t).
\end{equation}

In \cite{Pallard}, Pallard derived useful estimates on $\mathcal{I}_k(g;\sigma,t)$. 
However, unlike in \cite{Pallard}, we aim for estimates in terms of the function $\bar{W}(t)$ introduced above. Hence, we shall need some modifications of Pallard's estimates, which we present in the following lemma.

\begin{lemma} For $0\leq\sigma<t$, the integrals $\mathcal{I}_k(g;\sigma,t)$, $k=0,1$, satisfy 
 \begin{eqnarray}
\label{Integral 0}
\mathcal{I}_0(g;\sigma,t) & \leq & c\left\|g(\sigma)\right\|_{L^{\infty}_x}t\\ 
\mathcal{I}_1(g;\sigma,t) & \leq & c\frac{\left\|g(\sigma)\right\|_{L^2_x}}{\sqrt{t-\sigma}}\int^t_{\sigma}\left[1+\ln\bar{W}(s)\right]ds.
\label{Integral 1}
\end{eqnarray}
\end{lemma}
\begin{proof}
  The estimate (\ref{Integral 0}) is straightforward from (\ref{Pallard Integral}) with $k=0$. As for (\ref{Integral 1}), let us first rewrite the integral $\mathcal{I}_1$ in spherical coordinates
\begin{displaymath}
\mathcal{I}_1(g;\sigma,t)=\int^t_{\sigma}(s-\sigma)\int^{\pi}_0\int^{2\pi}_0g(\sigma,X(s)-\omega(s-\sigma))
\sin\phi d\theta d\phi ds,
\end{displaymath}
where $\omega=\omega(\theta,\phi)=(\cos\theta\sin\phi,\sin\theta\sin\phi,\cos\phi)$. As shown in \cite[Lemma 2.1]{Pallard}, condition (\ref{Condition Finite Momentum}) implies that the 
transformation $\pi_{\sigma},$ defined by 
\begin{displaymath}
   (s,\theta,\phi)\mapsto X(s)-\omega(s-\sigma)
\end{displaymath}
is a $C^1$-diffeomorphism whose Jacobian determinant has the form 
\begin{displaymath}
   J\pi_{\sigma}(s,\theta,\phi)=\left(\dot{X}(s)\cdot\omega-1\right)(s-\sigma)^2\sin\phi.
\end{displaymath}
Thus, the Cauchy-Schwarz inequality implies that
 \begin{eqnarray}
  \mathcal{I}_1(g;\sigma,t) & \leq & \left(\int^t_{\sigma}\int^{\pi}_0\int^{2\pi}_0g^2(\sigma,X(s)-\omega(s-\sigma))\left|J\pi_{\sigma}(s,\theta,\phi)\right|d\theta d\phi ds\right)^{\frac{1}{2}}\nonumber\\
                            &      & \times\left(\int^t_{\sigma}\int^{\pi}_0\int^{2\pi}_0\frac{\sin\phi d\theta d\phi ds}{1-\dot{X}(s)\cdot\omega}\right)^{\frac{1}{2}}\nonumber\\
                            & \leq & \left\|g(\sigma)\right\|_{L^2_x}\left(\int^t_{\sigma}\int^{\pi}_0\int^{2\pi}_0\frac{\sin\phi d\theta d\phi ds}{1-\dot{X}(s)\cdot\omega}\right)^{\frac{1}{2}}\nonumber.
 \end{eqnarray}
We estimate the integral
 \begin{eqnarray}
   \int^{\pi}_0\int^{2\pi}_0\frac{\sin\phi d\theta d\phi}{1-\dot{X}(s)\cdot\omega} & = 
   & \int^1_{-1}\frac{2\pi du}{1-|\dot{X}(s)|u}\nonumber\\
             & \leq & c\left(1+|\ln(1-\bar{v}(s))|\right)\nonumber\\
             & \leq & c\left(1+\ln\bar{W}(s)\right),\nonumber
 \end{eqnarray}
and therefore,
 \begin{eqnarray}
  \mathcal{I}_1(g;\sigma,t) & \leq &  c\left\|g(\sigma)\right\|_{L^2_x}\left(\int^t_{\sigma}
  \left[1+\ln\bar{W}(s)\right]ds\right)^{\frac{1}{2}}\nonumber\\
                            & \leq & c\frac{\left\|g(\sigma)\right\|_{L^2_x}}{\sqrt{t-\sigma}}
                            \int^t_{\sigma}\left[1+\ln\bar{W}(s)\right]ds.\nonumber
 \end{eqnarray}
In the final step we have used that $0<t-\sigma\leq\int_{\sigma}^t\left[1+\ln\bar{W}(s)\right]ds$. This completes the proof of the lemma.
\end{proof}

From (\ref{W Bounds h}) we have that $\left\|h(\sigma)\right\|_{L^{\infty}_x}\leq c_T\bar{W}(\sigma)$. 
Together with Lemma 1 this implies   
\begin{displaymath}
\left\|(h\bar{\left|K\right|})(\sigma)\right\|_{L^2_x}\leq \left\|h(\sigma)\right\|_{L^{\infty}_x}
\left\|\bar{\left|K\right|}(\sigma)\right\|_{L^2_x}\leq c_T\bar{W}(\sigma).
\end{displaymath}

Combining these inequalities and Lemma 2, we find that the integrals $I_0(h;t)$ and $I_1(h\bar{\left|K\right|};t)$ satisfy the estimates
\begin{eqnarray}
\label{Estimate I_0}
   I_0(h;t) & \leq & c_Tt\int^t_0\bar{W}(s)ds,\\ 
\label{Estimate I_1}
   I_1(h\bar{\left|K\right|};t) & \leq & c_T\int^t_0\int^t_{\sigma}\frac{\bar{W}(\sigma)}{\sqrt{t-\sigma}}\left[1+\ln\bar{W}(s)\right]dsd\sigma\\
   & = & c_T\int^t_0\int^s_0\frac{\bar{W}(\sigma)}{\sqrt{t-\sigma}}\left[1+\ln\bar{W}(s)\right]d\sigma ds\nonumber\\
   & \leq & c_T\sqrt{t}\int^t_0\bar{W}(s)\left[1+\ln\bar{W}(s)\right]ds\nonumber,
\end{eqnarray}
where in the last inequality we have used that $\bar{W}(t)$ is non-decreasing in $t$.

Finally, (\ref{Estimate I_0}) and (\ref{Estimate I_1}) combined with (\ref{Working Equation}) yield  
\begin{equation}
\label{The Most Important One}
   \bar{W}(t)\leq \bar{W}(0)+c_T+ c_T\int^t_0\bar{W}(s)\left[1+\ln\bar{W}(s)\right]ds. \end{equation}


\section{Proof of Theorem 1}

We recall the recursive method introduced in \cite{GS1} to generate a sequence of approximate solutions. A density argument allows to consider the following initial conditions:

Let $f_0\in C^2_0(\mathbb{R}^6)$, $(E_0,B_0)\in C^3_0(\mathbb{R}^3)$ and $(\partial_tE,\partial_tB)_0$ be in $C^2(\mathbb{R}^3)$. Let also $f^0(t,x,p)=f_0(x,p)$ and $(E^0,B^0)(t,x)=(E_0,B_0)(x)$.

Define $f^n$ as the solution of
\begin{eqnarray}
\label{Vlasov Recursive}
\partial_t f^n+v\cdot\nabla_{x}f^n+K^{n-1}\cdot\nabla_{p}f^n=0\hspace{2.cm}\\
f^n(0,x,p)=f_0(x,p),\hspace{4.3cm}\nonumber
\end{eqnarray}
(which is linear when $E^{n-1}$ and $B^{n-1}$ are given). Since $f_0\in C^2$, then $f^{n}$ is a $C^2$ function provided that $E^{n-1}$ and $B^{n-1}$ are $C^2$ as well. Morever, $f^{n}$ is constant along the characteristics of (\ref{Vlasov Recursive})
 \begin{eqnarray}
    \dot{X}_n(t,x_0,p_0) & = & V(P_n(t,x_0,p_0))\\
    \dot{P}_n(t,x_0,p_0) & = & K^{n-1}(t,X_n(t,x_0,p_0),P_n(t,x_0,p_0)).\nonumber
 \end{eqnarray}

Hence, for each $n$ we have that $f^{n}$ has compact support in $p$ provided that $E^{n-1}$ and $B^{n-1}$ are bounded functions. It follows that $\rho^{n}=\int f^{n}dp$ and $j^{n}=\int vf^{n}dp$ are well defined in $C^2$. Integrating (\ref{Vlasov Recursive}) over all $p\in\mathbb{R}^3$ yields the continuity equation $\partial_t\rho^n+\nabla_xj^n=0$. Then, we can define $(E^{n},B^{n})$ as the solution of  
\begin{eqnarray}
\label{Maxwell First Form}
\partial_tE^n-\nabla\times B^n & = & -j^n, \hspace{1.0cm} \nabla\cdot E^n=\rho^n\\
\partial_tB^n+\nabla\times E^n & = & 0, \hspace{1.45cm} \nabla\cdot B^n=0,\nonumber
\end{eqnarray}
which also solves the Maxwell equations in the second order form
 \begin{eqnarray}
 \label{Wave Recursive}
     \left(\partial^2_t-\Delta\right)E^{n} & = &-\nabla\rho^{n}-\partial_tj^{n}\\
 \left(\partial^2_t-\Delta\right)B^{n} & = & \nabla\times j^{n}\nonumber
 \end{eqnarray}
with the given Cauchy data.
\begin{lemma}[Glassey \& Strauss] 
  If $f^{n}\in C^2$ solves (\ref{Vlasov Recursive}), and $(E^{n},B^{n})$ solves (\ref{Wave Recursive}), then $E^{n}$ and $B^{n}$ are $C^2$ bounded functions. 
\end{lemma}
\begin{proof} See \cite[Lemma 7]{GS1}.
\end{proof}

Now, let $\bar{p}_n(t)=\sup\left\{\left|p\right|:\exists 0\leq s\leq t, x\in\mathbb{R}^3:f^{n}(s,x,p)\neq0\right\}$ and define $\bar{v}_n(t)$ and $\bar{W}_n(t)$ according to (\ref{Definitions}). Since $f^n$ has compact support in $p$, then (\ref{Condition Finite Momentum}) holds for each $n\in\mathbb{N}$. On the other hand, the assumption of Theorem 1 is
\begin{displaymath}
   \sup_{0\leq t\leq T}\sup_{x\in\mathbb{R}^3}\int_{\mathbb{R}^3}f^{n}(t,x,p)dp\leq c_T 
\end{displaymath} 
uniformly in $n$. Then $(f^{n},E^{n},B^{n})$ satisfies the two assumptions made in the previous section. In particular, (\ref{Condition Finite Momentum}) holds for each $n$ and (\ref{Assumption 2}) uniformly in $n$. In addition, by using the equations in (\ref{Maxwell First Form}) and noticing that $\left\|f^n(t)\right\|_{L^1_{x,p}}=\left\|f_0\right\|_{L^1_{x,p}}$ for all $0\leq t\leq T$ and $n\in\mathbb{N}$, the proof of Lemma \ref{Poynting} applies in exactly the same way for the approximate solutions, so the estimate (\ref{Assumption 3}) holds uniformly in $n$ as well. By virtue of (\ref{Estimate Characteristics Like}) we also have 
\begin{displaymath}
  W_n(t)\leq\ W_n(0)+\int^t_0\left| E^{n-1}(s,X_n(s))\right|ds
\end{displaymath}
for each $0\leq t\leq T$. Hence, recalling that $\bar{W}_n(t)=\sqrt{1+\bar{p}^2_n(t)}$, it is not difficult to check that (\ref{The Most Important One}) becomes
\begin{displaymath}
   \bar{W}_n(t)\leq \bar{W}_0+c_T+c_T\int^t_0\bar{W}_{n-1}(s)\left[1+\ln\bar{W}_n(s)\right]ds, 
\end{displaymath}
where $\bar{W}_0\equiv\bar{W}_n(0)$. If we set $$\tilde{W}_n(t)= \sup_{k\leq n}\bar{W}_k(t)$$ we deduce the logarithmic Gronwall inequality 
 \begin{equation}
 \tilde{W}_{n}(t)\leq\bar{W}_0+c_T+c_T\int^t_0\tilde{W}_{n}(s)\left[1+\ln\tilde{W}_{n}(s)
 \right]ds.\nonumber 
 \end{equation}
Then, owing to the compact $p$-support of $f_0$, we obtain the uniform bound
$$\sup_{0\leq t\leq T}\sup_{n\in\mathbb{N}}\bar{p}_n(t)\leq\sup_{0\leq t\leq T}\sup_{n\in\mathbb{N}}\tilde{W}_{n}(t)\leq c_T
$$ 
with a constant $c_T$ not depending on $n$. But this is the assumption made in \cite[Theorem 1]{GS1}. Therefore, the sequence $\left\{(f^{n},E^{n},B^{n})\right\}$ converges uniformly in $t\in[0,T]$, $x\in\mathbb{R}^3$ and $p\in\mathbb{R}^3$ to the unique classical solution of the RVM system. This concludes the proof of the theorem.

\end{document}